\documentclass[11pt]{article} %full
\usepackage{fullpage} %full
\usepackage{amsfonts} %full
\usepackage{algorithmic}
\usepackage{algorithm}
\usepackage{amsmath,amsthm,amssymb} %full
\usepackage{times} %full
\usepackage{latexsym}
\usepackage{epic}
\usepackage{epsfig}
\usepackage{amscd}
\usepackage{url}
\usepackage{verbatim}

\newtheorem{theorem}{Theorem}[section]

\newtheorem{lemma}[theorem]{Lemma}

\newtheorem{claim}[theorem]{Claim}
\newtheorem{definition}[theorem]{Definition}
\newtheorem{remark}[theorem]{Remark}

\newtheorem{fact}[theorem]{Fact}

\newtheorem{thm}[theorem]{Theorem}

\newcommand{\by}{\times}

\newcommand{\set}[1]{\left\{ #1 \right\}}
\newcommand{\union}{\cup}
\newcommand{\intersect}{\cap}

\newcommand{\sm}{\setminus}

\renewcommand{\tilde}{\widetilde}
\renewcommand{\bar}{\overline}

\DeclareMathOperator{\poly}{poly}

\def\pr{\qopname\relax n{\mathbf{Pr}}}
\def\ex{\qopname\relax n{\mathbf{E}}}
\def\min{\qopname\relax n{min}}
\def\max{\qopname\relax n{max}}

\newcommand{\expect}[2][]{\ex_{#1} [#2]}

\newcommand{\RR}{\mathbb{R}}
\newcommand{\RRp}{\RR_+}

\newcommand{\NN}{\mathbb{N}}

\def\A{\mathcal{A}}

\def\D{\mathcal{D}}

\def\H{\mathcal{H}}
\def\I{\mathcal{I}}

\def\P{\mathcal{P}}

\def\V{\mathcal{V}}
\def\X{\mathcal{X}}
\def\Y{\mathcal{Y}}

\def\sse{\subseteq}

\newcommand{\grad}{\bigtriangledown}

\newcommand{\eat}[1]{}

\newcommand{\INPUT}{\item[\textbf{Input:}]}
\newcommand{\OUTPUT}{\item[\textbf{Output:}]}
\newcommand{\PARAMETER}{\item[\textbf{Parameter:}]}

\newcommand{\maxi}[1]{\mbox{maximize} & {#1 } & \\}
\newcommand{\st}{\mbox{subject to} }
\newcommand{\con}[1]{&#1 & \\}
\newcommand{\qcon}[2]{&#1, & \mbox{for } #2.  \\}
\newenvironment{lp}{\begin{equation}  \begin{array}{lll}}{\end{array}\end{equation}}
\newenvironment{lp*}{\begin{equation*}  \begin{array}{lll}}{\end{array}\end{equation*}}

\title{A Truthful Randomized Mechanism for Combinatorial Public Projects via Convex Optimization\footnote{Extended abstract appears in \emph{Proceedings of the 12th ACM Conference on Electronic Commerce (EC), 2011.}}}

\author{Shaddin Dughmi\thanks{Supported by NSF Grant CCF-0448664.} \\
Department of Computer Science\\
Stanford University\\
{\tt shaddin@cs.stanford.edu}}
\begin{document}

%beginproc
% \conferenceinfo{EC'11,} {June 5--9, 2011, San Jose, California, USA.} 
% \CopyrightYear{2011} 
% \crdata{978-1-4503-0261-6/11/06} 
% \clubpenalty=10000 
% \widowpenalty = 10000
%endproc 
\maketitle

\begin{abstract}
In \emph{Combinatorial Public Projects}, there is a set of projects that may be undertaken, and a set of self-interested players with a stake in the set of projects chosen. A public planner must choose a subset of these projects, subject to a resource constraint, with the goal of maximizing social welfare. Combinatorial Public Projects has emerged as one of the paradigmatic problems in \emph{Algorithmic Mechanism Design}, a field concerned with solving fundamental resource allocation problems in the presence of both selfish behavior and the computational constraint of polynomial-time.

We design a polynomial-time, truthful-in-expectation,
$(1-1/e)$-approximation mechanism for welfare maximization in a fundamental variant of combinatorial public projects.  Our results apply to combinatorial public projects when
players have valuations that are {\em matroid rank sums (MRS)}, which 
encompass most concrete examples of submodular functions studied in 
this context, including coverage functions, matroid weighted-rank
functions, and convex combinations thereof.  Our approximation factor is the best possible,  assuming $P \neq NP$.
Ours is the first mechanism that achieves a constant factor approximation for a natural NP-hard variant of combinatorial public projects.

%Todo in full version: say it is an instantiation of convex rounding?

% Our mechanism is an instantiation of  the framework of Dughmi, Roughgarden and Yan \cite{DRY11}.  instantiation of a new framework for designing
% approximation mechanisms based on randomized
% rounding algorithms.  A typical such algorithm first optimizes over
% a fractional relaxation of the original problem, and then randomly
% rounds the fractional solution to an integral one.  With rare
% exceptions, such algorithms cannot be converted into
% truthful mechanisms.  
% The high-level idea of our mechanism design framework is to optimize
% {\em directly over the (random) output of the rounding algorithm},
% rather than over the {\em input} to the rounding algorithm.
% This approach leads to truthful-in-expectation mechanisms, and these
% mechanisms can be implemented efficiently when the corresponding
% objective function is concave.  For bidders with MRS valuations, we
% give a novel randomized rounding algorithm that leads to both
% a concave objective function and a $(1-1/e)$-approximation of the
% optimal welfare.

%  Our approximation mechanism also provides an interesting separation
% between the power of maximal-in-distributional-range mechanisms and
% that of universally truthful (or deterministic) VCG-based mechanisms,
% which cannot achieve a constant-factor approximation for welfare
% maximization with MRS valuations.

\end{abstract}

\thispagestyle{empty} %full
\addtocounter{page}{-1} %full

 %proc
% \category{F.0}{Theory of Computation}{General}

% \terms{Algorithms, Economics, Theory}

% %\keywords{Algorithmic Mechanism Design, Truthfulness}
 %endproc

\newpage

\section{Introduction}
% Algorithmic Mechanism Design and CPP, mention flexible (footnote?) and shahar

The overarching goal of {\em algorithmic mechanism design} is to design
computationally efficient algorithms that solve or approximate
fundamental resource allocation problems in which the underlying data is a
priori unknown to the algorithm.  A problem that has received much attention in this context --- albeit mostly in the form of negative results --- is \emph{Combinatorial Public Projects} (CPP).  Here,
there are $m$  \emph{projects} being considered by a public planner, $n$ \emph{players}, and a bound $k\leq m$ on the number of projects that may be chosen.  Each player $i$ has a private \emph{valuation} $v_i(S)$ for each subset $S$ of the projects.  We consider the \emph{flexible} variant of CPP, where a feasible solution is a set of at most $k$ projects\footnote{This is in contrast to the \emph{exact} variant, where each feasible solution consists of \emph{exactly} $k$ projects --- a difference that is uninteresting in an approximation algorithms context, yet has major implications when incentives are in the picture. For more on the distinction between the two variants, we refer the reader to \cite{Dobzin11}.}. The goal is to choose a feasible set of projects $S$ maximizing \emph{social welfare}:  $\sum_i v_i(S)$. The valuations are initially unknown to the public planner, and must be elicited from the (self-interested) players. A ``mechanism'' for CPP extracts this information, and decides on a set of projects to undertake. The mechanisms we consider can charge the players payments in order to incentivize truthful reporting of their valuations. Moreover, we seek mechanisms that run in polynomial time. 

Since CPP is highly inapproximable for general valuations --- even by non-truthful algorithms --- it is most interesting to study CPP for restricted classes of valuations. Most notable among these are submodular valuations, as they naturally model the pervasive notion of ``diminishing marginal returns''.   In this paper, we study CPP for a fundamental and large subset of submodular valuations: \emph{Matroid Rank Sum Valuations}. This class includes most concrete examples of submodular functions studied in this context. Most notably, it includes the canonical and arguably most natural example of submodularity: coverage functions.

% Context: Welfare maximization problems: canonical: CA and CPP
Combinatorial public projects and its variants are examples  of \emph{welfare maximization problems}. There are many other examples, most notable among them are \emph{combinatorial auctions}, with their many variants (see e.g. \cite{NRTV07}). Welfare maximization problems occupy a central position in mechanism design, not only because of the fundamental nature of the utilitarian objective, but also due to the rich economic theory surrounding them. Most notably, the celebrated Vickrey-Clarke-Groves (VCG)  mechanism (see e.g. \cite{NRTV07}) is a general solution for all these problems, at least from an economic perspective. The VCG mechanism is \emph{truthful}, in that it is in a player's best interest to report his true valuations regardless of the reports of the other players. Moreover, VCG finds the welfare maximizing solution. 

Unfortunately, however, most interesting welfare maximization problems, such as combinatorial public projects, are NP-hard. Therefore, implementing VCG  efficiently --- i.e. in polynomial time --- is impossible unless $P = NP$. Moreover, as first argued in \cite{NR00}, most existing approximation algorithms --- unlike exact algorithms --- cannot be converted to truthful mechanisms by the imposition of a suitable payment scheme. This necessitates the design of carefully crafted approximation algorithms, tailored specifically for truthfulness. Understanding the power of these  truthful approximation mechanisms is the central goal of algorithmic mechanism design. This research agenda was first advocated by Nisan and Ronen \cite{NR99}. Since then, combinatorial auctions and combinatorial public projects have emerged as the paradigmatic ``challenge-problems'' of the field, with much work in recent years establishing upper and lower-bounds on truthful polynomial-time mechanisms for these problems, for example: \cite{LS05,DNS05,DS06,DNS06,DN07a,D07,DS08,PSS08,BDFKMPSSU10,BSS10,Dobzin11,DRY11}.

% Lowerbounds: In general, but focus on CPP being a a tough one... 
The ``holy grail'' of algorithmic mechanism design is to design polynomial-time truthful approximation mechanisms that match the approximation guarantee of the best (non-truthful) polynomial-time approximation algorithm. Unfortunately, several recent impossibility results have shed serious doubt on the possibility of this goal \cite{DN07a,PSS08,BDFKMPSSU10,BSS10,Dobzin11}. Combinatorial public projects, in particular, bore the brunt of the most brutal of these negative results \cite{PSS08,BSS10,Dobzin11}.  Fortunately,  all but one of these lower bounds apply exclusively to deterministic mechanisms, and none apply to randomized mechanisms for the --- arguably more natural --- flexible variant of combinatorial public projects.

% Upperbounds: deterministic fail, MIDR algorithms emerging. At first, for "simpler" problems such as MUA and problems with an FPTAS, and problems with LP relaxations satisfying a condition
As the limitations of deterministic mechanisms became apparent, a recent research direction has focused on designing randomized approximation mechanisms for the fundamental problems of algorithmic mechanism design \cite{LS05,DD09,DR10,DFK11,DRY11}. These mechanisms are instances of the only general approach\footnote{The random sampling approach used in \cite{D07}, while arguably  general, does not seem applicable beyond auction settings --- in particular, it is not applicable to combinatorial public projects.} known for designing
(randomized) truthful mechanisms:  via {\em maximal-in-distributional
  range (MIDR) algorithms}~\cite{DD09}. 
An MIDR algorithm fixes a set of distributions over feasible solutions --- the {\em distributional range} --- independently of the 
valuations reported by the self-interested participants, and outputs a
random sample from the distribution that maximizes expected (reported)
welfare.
The ``Vickrey-Clarke-Groves (VCG)'' payment scheme renders an MIDR
algorithm \emph{truthful-in-expectation} --- that is, a player unaware of the coin flips of the mechanism maximizes his expected utility by reporting truthfully.

% Recent push: DRY design a mechanism for CA with MRS valuations based on a framework, does this framework apply to other problems?
Recently Dughmi, Roughgarden and Yan \cite{DRY11} presented the most general framework to date for the design of maximal-in-distributional-range algorithms. Their approach is based on convex optimization, and generalizes the celebrated linear-programming based approach of Lavi and Swamy \cite{LS05}. Given a mathematical relaxation to a welfare maximization problem, \cite{DRY11} advocates designing randomized rounding schemes that are \emph{convex}. Given a convex rounding scheme, the problem of finding the best \emph{output} of the rounding scheme is a convex optimization problem solvable in polynomial time, and implements an MIDR allocation rule. They then show how to design a convex rounding scheme for combinatorial auctions with matroid rank sum valuations, yielding an optimal $(1-1/e)$ approximation mechanism. We elaborate on the framework of \cite{DRY11} in Section \ref{sec:convexrounding}. 

 By reducing the problem of designing a truthful mechanism to that of designing a convex rounding scheme, the approach of \cite{DRY11} yielded the first optimal truthful mechanism for a variant of combinatorial auctions with restricted valuations. It is now natural to wonder if their approach is applicable to other welfare maximization problems. In particular, can the convex rounding framework be used to obtain optimal approximation mechanisms for interesting variants of Combinatorial Public Projects? 

We answer this question in the affirmative, and elaborate on our contributions below.

\subsection{Contributions}

%Yes!! We use their framework --- in more advanced fashion --- to get a similar result for CPP. The constraints of CPP make it so that the rounding must be correlated, and hence trickier to get a rounding that is both convex and a good approximation. We detail contributions below

 We design a $(1-1/e)$-approximate convex rounding scheme for combinatorial public projects with matroid rank sum valuations. This yields a $(1-1/e)$-approximate truthful-in-expectation mechanism for CPP, running in expected polynomial-time. This is the best approximation possible for this problem, even without truthfulness, unless $P=NP$. Therefore, ours is the first truthful mechanism for an NP-hard variant of CPP that matches the approximation ratio of the best non-truthful algorithm. Our results works with ``black-box'' valuations, provided that players can answer a randomized analogue of value oracles.

To prove our results, we follow the general outline of \cite{DRY11}. However, our task is more challenging: whereas in combinatorial auctions, randomized rounding may allocate each item independently (the approach taken in \cite{DRY11}), this is not possible in CPP.  We must respect the cardinality constraint of $k$ on the set of chosen projects, and therefore our rounding scheme must by fiat be \emph{dependent}. This presents a major challenge in analyzing our rounding scheme. Whereas the expected value of a submodular function on a product distribution (i.e. independent rounding) has been studied extensively, and is closely related to the now well-understood multi-linear  (see e.g. \cite{CCPV07,V08}), analyzing the expected value of a dependent distribution --- in particular proving it to be a concave function of underlying parameters --- is a technical challenge that we overcome by combining techniques from combinatorics, convex analysis, and matroid theory.

%We design 1-1/e for CPP(MRS). Optimal assuming P\neq NP. First optimal truthful approximation for variant of CPP with restricted valuations. Our results are in a lottery-query and communication complexity model -- this is without loss for some succinctly represented classes of vals, and can be removed in case of epsilon-truthfulness.

%We follow convex rounding framework. We design a new "correlated" rounding scheme designed to respect cardinality bound of CPP. Analyzing this scheme for MRS valuations is more involved than in DRY, because of the dependence between the items.

\subsection{Additional Related Work}
Combinatorial Public Projects, in particular its \emph{exact} variant, was first introduced by Papadimitriou, Schapira and Singer \cite{PSS08}. They show that no deterministic truthful mechanism for exact CPP with submodular valuations can guarantee better than a $O(\sqrt{m})$ approximation to the optimal social welfare. The  non-strategic version of the problem, on the other hand, is equivalent to maximizing a submodular function subject to a cardinality constraint, and admits a $(1-1/e)$-approximation algorithm due to Nemhauser, Wolsey and Fisher \cite{nwf78}, and this is optimal \cite{RazS97} assuming $P \neq NP$.  

Buchfuhrer, Schapira and Singer \cite{BSS10}  explored approximation algorithms and truthful mechanisms for CPP with various classes of valuations in the submodular hierarchy. The most relevant result of \cite{BSS10} to our paper is a lower-bound of $O(\sqrt{m})$ on \emph{deterministic} truthful mechanisms for the exact variant of CPP with coverage valuations --- a class of valuations for which our \emph{randomized} mechanism for flexible CPP obtains a $(1-1/e)$ approximation.

Most recently, Dobzinski \cite{Dobzin11} showed two lower bounds for CPP in the value oracle model:  A lower bound of $O(\sqrt{m})$ on universally truthful mechanisms for flexible CPP with submodular valuations, and a lower bound of $O(\sqrt{m})$ on truthful-in-expectation mechanisms for \emph{exact} CPP with submodular valuations. We note that the latter was the first unconditional lower bound on truthful-in-expectation mechanisms.

%CPP Papers by Yaron et al.. NOTe that the computation and incentives paper has an ub and lb for coverage!! Definitely cite!!!

%all related work in DRY + DRY + Shahar's recent stuff, particularly on CPP

\section{Preliminaries}
\def\S{\mathcal{S}}

\subsection{Combinatorial Public Projects}

In \emph{Combinatorial Public Projects} there is a set
$[m]= \set{1,\ldots,m}$ of \emph{projects}, a cardinality bound $k$ such that $0 \leq k \leq m$, and a set
$[n]=\set{1,\ldots,n}$ of \emph{players}. Each player $i$ has a valuation
function $v_i:2^{[m]}\rightarrow \RRp$ that is normalized
($v_i(\emptyset) = 0$) and monotone ($v_i(A) \leq v_i(B)$ whenever
$A \sse B$). In this paper, we consider the \emph{flexible} variant of combinatorial public projects: a feasible solution is a set $S \sse [m]$ of projects with $|S| \leq k$. Player $i$'s value for outcome $S$ is equal to $v_i(S)$.  The
goal is to choose the feasible set $S$ maximizing \emph{social welfare}:
$\sum_i v_i(S)$.  

We consider Combinatorial Public Projects where each player's valuation $v_i$ is know to lie in some set $\V$ of valuation functions. We abbreviate the set of instances of  CPP constrained to valuations $\V$ as CPP($\V$). As first defined in \cite{PSS08}, CPP was considered with $\V$ equal to the set of monotone submodular functions. In this paper, we focus on CPP with matroid-rank-sum (MRS) valuations --- a large subset of monotone submodular functions.

% \subsection{Optimization Problems}
% \label{sec:problem}

% We consider  optimization problems $\Pi$ of the following general
% form. Each instance of $\Pi$ consists of a \emph{feasible set}  $\S$, and a
% \emph{objective function} $v:\S \to \RR$.  
% %We make no assumptions
% %about $v$ --- it may be nonlinear, nonsubmodular, etc. 
% %We are motivated
% %by problems for which $v(x)$ is the social welfare of outcome $x$. 
% %We do not assume anything about the representation of $\S$ and $v$ --- the
% %representation may be implicit. 
% The solution to an instance of $\Pi$
% is given by the following optimization problem. 

% \begin{lp}\label{lp:problem}
%   \maxi{w(x)}
% \st
% \con{x \in \S.}
% \end{lp}

\subsection{Mechanism Design Basics}\label{sec:MD}

We consider direct-revelation mechanisms for combinatorial public projects. Fix $m$,$n$, and $k$, and let $\S=\{S \sse [m] : |S| \leq k\}$ denote the set of feasible solutions.  A mechanism comprises an {\em
  allocation rule}, which is a function from (hopefully
truthfully) reported valuation functions $v_1,\ldots,v_n:2^{[m]} \to \RR$ to a feasible outcome
$S \in \S$, and a {\em payment rule}, which is a function from
reported valuation functions to a required payment from each player.
We allow the allocation and payment rules to be randomized.

A mechanism with allocation and payment rules $\A$ and $p$ is {\em
  truthful-in-expectation} if every player always maximizes its expected
  payoff by truthfully reporting its valuation function, meaning that
\begin{equation}\label{eq:truthful}
\ex[v_i(\A(v)) - p_i(v)] \geq \ex[ v_i(\A(v'_i,v_{-i})) -
  p_i(v'_i,v_{-i})]
\end{equation}
for every player~$i$, (true) valuation function~$v_i$, (reported)
valuation function~$v'_i$, and (reported) valuation functions~$v_{-i}$ of
the other players.
The expectation in~\eqref{eq:truthful} is over the coin flips of the
mechanism.  
%If the mechanism is deterministic and satisfies this condition, then
%it is simply called {\em truthful}.

The mechanisms that we design can be thought of as randomized
variations on the classical VCG mechanism, as we explain next.  
Recall that the {\em VCG
  mechanism} is defined by the (generally intractable)
allocation rule that selects the welfare-maximizing outcome with
respect to the reported valuation functions, and the payment rule that
charges each player~$i$ a bid-independent ``pivot term'' minus the
reported welfare earned by other players in the selected outcome.  This
(deterministic) mechanism is truthful; see e.g.~\cite{Nis07}.

Now let $dist(\S)$ denote the probability distributions over the feasible set $\S$, and let $\D \sse dist(\S)$ be a compact subset of them.
The corresponding {\em Maximal-In-Distributional-Range (MIDR)} allocation rule is defined as follows: given reported valuation
functions $v_1,\ldots,v_n$,  return an outcome that is sampled randomly
from a distribution $D^* \in \D$ that maximizes the expected welfare
$\ex_{S \sim D}[\sum_i v_i(S)]$ over all distributions $D \in \D$.
Analogous to the VCG mechanism, there is a (randomized) payment rule that
can be coupled with this allocation rule to yield a
truthful-in-expectation mechanism (see~\cite{DD09}).

\subsection{Matroid Rank Sum Valuations}\label{sec:MRS}
%In order to define MRS valuations, we need to first recall basic
%concepts from matroid theory. An overview of the basic concepts we use
%here, such as the definition of matroids and their rank functions, is
We now define matroid rank sum valuations. Relevant concepts from matroid theory are reviewed in Appendix  \ref{app:matroids}. 
\begin{definition}\label{def:mrs}
A set function $v:2^{[m]} \to \RR$ is a \emph{matroid rank sum (MRS)}
function if there exists a family of matroid rank functions
$u_1,\ldots,u_\kappa: 2^{[m]} \to \RR$, and associated non-negative
weights $w_1,\ldots,w_\kappa \in \RR^+$, such that $v(S) =
\sum_{\ell=1}^\kappa w_\ell u_\ell(S)$ for all $S \sse [m]$. 
\end{definition}
%
%If we view a set function on ground set $[m]$ as a vector indexed by
%sets $S \in 2^{[m]}$, it is easy to see that MRS valuations form a
%\emph{convex cone}. Our results rely on the geometric properties of
%this cone, and not on how we represent its member functions. 
We do not assume any particular representation of MRS functions, and
require only oracle access to 
their (expected) values on certain distributions (see Section
\ref{sec:lotteryval}). 
%In particular, the results of Section
%\ref{sec:CA} are agnostic to the number of matroids $\kappa$ in the
%weighted combination above -- it may be arbitrarily large, even
%infinite as far as we care\footnote{We do note, however, that there is
%  a finite -- though doubly exponential  -- number of matroids on any
%  finite ground set. Therefore, every MRS function can theoretically
%  be written as a finite weighted sum of matroid rank functions.}.  
 MRS valuations include most concrete examples of monotone submodular
functions that appear in the literature --- this includes
coverage functions\footnote{A coverage function $f$ on ground set
  $[m]$ designates some set $\Y$, and $m$ subsets $A_1,\ldots,A_m \sse
  \Y$, such that $f(S) = | \union_{\ell \in S} A_\ell|$. We note that
  $\Y$ may be an infinite, yet measurable, space. Coverage functions
  are arguably {\em the} canonical example of a submodular
  function.}, matroid weighted-rank
functions\footnote{This is a generalization of matroid rank functions,
  where weights are placed on elements of the matroid. It is true,
  though not immediately obvious, that a matroid weighted-rank function can be
  expressed as a weighted combination of matroid (unweighted) rank
  functions --- see e.g. \cite{revenuesubmod}.}, and all convex
combinations thereof.  Moreover, as shown in \cite{RazS97}, $1-1/e$
is the best approximation possible for CPP with coverage valuations --- and hence also for MRS valuations --- in polynomial
time, even ignoring strategic considerations. That being said, we note that some interesting submodular functions  --- such as some budget additive functions\footnote{  A set function $f$ on ground set $[m]$  is \emph{budgeted additive} if
  there exists a constant $B\geq 0$ (the budget) such that $f(S) =
  \min(B,\sum_{j \in S} f(\set{j}))$.
} --- are not in the matroid rank sum family.

\subsection{Lotteries and Oracles}
\label{sec:lotteryval}
A \emph{value oracle} for a valuation $v: 2^{[m]} \to \RR$ takes as
input a set $S \sse [m]$, and returns $v(S)$.  We define an analogous
oracle that takes in a description of a simple lottery 
% (where each
%good is included independently with some probability) 
over sets $S 
\sse [m]$, and outputs the expectation of $v$ over this lottery. %The lotteries we consider will be of a very simple form, which we describe next.%: $k$ projects are drawn with replacement from  a probability distribution over projects.
%In
%particular, we use the simplest possible such lotteries: each item is
%included in $S$ independently with some probability. Given a vector
%$x\in [0,1]^m$, let $D_x$ be the distribution over sets $S$ that
%includes $j \in S$ independently with probability $x_j$. 

Let $k \in [m]$, let $R \sse [m]$, and let $x\in [0,1]^m$ be a vector such that $\sum_j x_j \leq 1$. We interpret $x$ as a probability distribution over $[m] \union \set{*}$, where $*$  represents not choosing a project.  Specifically, project $j\in [m]$ is chosen with probability $x_j$, and $*$ is chosen  with probability $1-\sum_j x_j$.  We define a distribution $D^R_k(x)$ over $2^{[m]}$, and call this distribution the \emph{$k$-bounded lottery with marginals $x$ and promise $R$}. We sample $S \sim D^R_k(x)$ as follows: Let $j_1,\ldots, j_k$ be independent draws from $x$, and let $S= R \union \set{j_1,\ldots,j_k} \sm \set{*}$. Essentially, this lottery commits to choosing projects $R$, and adds an additional $k$ projects chosen randomly with replacement from distribution $x$. When $R=\emptyset$, as will be the case through most of this paper, we omit mention of the promised set. We can now define a randomized analogue of a value oracle that returns the expected value of a bounded-lottery.  %We use $G^v_k(x)$ to denote the expected value of $v(S)$ over draws $S \sim
%D_k(x)$ from this lottery.  
\begin{definition}\label{def:bounded_lottery_oracle}
  A \emph{bounded-lottery-value oracle} for set function $v: 2^{[m]} \to \RR$
  takes as input a vector $x \in [0,1]^m$ with $\sum_j x_j \leq 1$, a bound $k \in [m]$, and a set $R \sse [m]$, and outputs $\expect[S \sim D^R_k(x)]{v(S)}$.
%  \begin{equation}\label{eq:G} G^v_k(x) =\expect[S \sim D_k(x)]{v(S)} 
%=  \sum_S v(S) \prod_{j \in S} x_j \prod_{j \neq S}  (1-x_j).
%\end{equation}  
\end{definition}
%Viewed differently, the lottery-value oracle simply evaluates the
%\emph{multi-linear extension} $F_v$ of $v$ (see
%\cite{vondrak_thesis}).  Operating in the lottery-value oracle model
%is tantamount to assuming that players can efficiently ``introspect''
%in order to discover their value for a simple lottery where each item
%is assigned independently. While lottery-value oracles are a natural
%extension of value oracles to randomized allocations,  they cannot in
%general be simulated efficiently by their more traditional
%counterparts. They can, however, be approximated arbitrarily well
%using value oracles by random sampling, as shown in
%\cite{V08}. Unfortunately, the stringent requirements of MIDR
%mechanisms -- namely, that we output \emph{exactly} optimal solutions
%-- do not seem to tolerate such sampling error, however small. This
%motivates our definition of lottery-value oracles.

In our model for CPP, we assume that a player with valuation function $v_i$ can answer bounded-lottery-value oracle queries for $v_i$. A bounded-lottery-value oracle is a generalization of value oracles. Nevertheless, it is the case that a bounded-lottery-value oracle can be implemented using a value oracle for some succinctly 
represented examples of MRS valuations, such as explicit coverage
functions (In similar fashion to \cite[Appendix A]{DRY11}). 

%We now define our \emph{k-lottery-value-oracle model} for combinatorial public projects. We assume that a player with valuation function $v_i:2^{[m]} \to \RR$ can, for each $k \in [m]$, answer $k$-lottery-value-oracle queries for $v_i$. We also assume the player can answer the more traditional value-oracle queries for $v_i$. Moreover, for technical reasons that will  become apparent in Appendix~\ref{sec:solveconvex}\footnote{We believe we can remove these technical requirements in the full version of the paper, removing the need for these additional types of queries.}, we require a player to answer an additional type of simple query. For a project $j \in [m]$, let $v_i^j(S) = v_i(S \union \set{j})$ be the valuation of player $i$ if he is ``promised'' project $j$ up-front. For each project $j$ and $k \in [m]$, we require that player $i$ be able to answer $k$-lottery-value-oracle queries for $v_i^j$ -- i.e. report his expected value for a $k$-lottery $D_k(x)$ if he is additionally promised project $j$ up-front. These types of queries are again easy to implement for various succinctly represented examples of MRS valuations, like explicity represented coverage functions.

More generally we note that  bounded-lottery-value oracles can  be approximated arbitrarily well, with high probability, using value oracles; this is done by random sampling, and we omit the technical details. Unfortunately, we are not able  to reconcile the incurred sampling errors --- small as they may be --- with the requirement that our mechanism be \emph{exactly} truthful. We suspect that relaxing our solution concept to approximate truthfulness -- also known as $\epsilon$-truthfulness -- would remove this difficulty, and allow us to relax our oracle model to the more traditional value oracles.

\subsection{Convex Rounding}
\label{sec:convexrounding}
In this section, we review \emph{convex rounding}, a framework for the design of truthful mechanisms introduced by Dughmi, Roughgarden and Yan~\cite{DRY11}. We present the main definitions and lemmas  as they pertain to combinatorial public projects. For a more thorough and general treatment of convex rounding, we refer the reader to \cite[Section 3]{DRY11}.

We consider the standard integer programming formulation of CPP. There is a variable $x_j\in \set{0,1}$ for each project $j\in[m]$, and the goal is to set at most $k$ of the variables to $1$ so that the welfare $v(x)=\sum_i v_i(\set{j : x_j = 1})$ is maximized. We \emph{relax} this integer program in the obvious way to the polytope $\P = \set{x \in \RR^m : \sum_j x_j \leq k, x \succeq 0}$. We postulate a \emph{rounding scheme} $r$ that maps  points of $\P$ to the feasible solutions $\S=\set{S\sse [m] : |S| \leq k}$ of CPP. We allow $r$ to be randomized, so that $r(x)$ is a distribution over $\S$ for each $x \in \P$. 
%in full, re-insert relax this integer program _in the obvious way_

Traditionally, approximation algorithms optimize an objective $\tilde{v}(x)$ --- often a simple extension of $v$ to $\P$ --- over the set $\P$ of fractional solutions, and then round the optimal fractional point $x^*$ to a solution $r(x^*)$ in the original feasible set $\S$. Many of the best approximation algorithms for various problems are based on this relax-solve-round framework. Unfortunately, however, this approach is almost always incompatible with the design of truthful mechanisms, due to the fact that the rounding step is often unpredictable. Truthful mechanism design, on the other hand, is intimately tied to \emph{exact optimization}, as evidenced by the fact that the vast majority truthful mechanisms for multi-parameter problems are based on the VCG paradigm (see Section \ref{sec:MD}).

In an effort to reconcile the techniques of approximation algorithms and truthful mechanism design, Dughmi, Roughgarden and Yan proposed \emph{optimizing directly on the output of the rounding scheme, rather than on its input}. This defines an optimization problem induced by relaxation $\P$ and rounding scheme $r$. Stated for CPP with the relaxation as described above, the problem is as follows.
\begin{lp}
\label{lp:absorbrounding}
  \maxi{\ex_{S \sim r(x)}[ \sum_i v_i(S)]}
  \st
  \con{\sum_{j=1}^m x_j \leq k}
  \qcon{0 \leq x_j \leq 1}{j=1,\ldots,m}
\end{lp}
They consider a simple allocation rule, which we state for CPP in Algorithm \ref{alg:midr}, that solves \eqref{lp:absorbrounding} optimally. They observe that this allocation rule is maximal-in-distributional-range.
\begin{algorithm}
\caption{MIDR Allocation Rule for CPP}
\label{alg:midr}
\begin{algorithmic}[1]
\PARAMETER $n$,$m$,$k$
\PARAMETER (Randomized) rounding scheme $r$
\INPUT Valuation functions $\set{v_i}_{i=1}^n$
\OUTPUT  A set $S \sse [m]$ with $|S| \leq k$
\STATE Let $x^*$ be an optimal solution to \eqref{lp:absorbrounding}
\STATE Let $S \sim r(x^*)$ \label{algstep:sample}
 \end{algorithmic}
\end{algorithm}
\begin{lemma}[\cite{DRY11}]
 \label{lem:MIDR}
  Algorithm \ref{alg:midr} is an MIDR allocation rule.
\end{lemma}

For $\alpha \leq 1$, we say that the rounding scheme $r$ for CPP($\V$) is
\emph{$\alpha$-approximate} if, whenever $x$ is an integer point of $\P$ corresponding to a set $S \in \S$, and $v_i \in \V$ for each $i$, we have that $\ex_{T \sim r(x)}[\sum_i v_i(T)] \geq \alpha \sum_i v_i(S)$. In other words, rounding does not degrade the quality of an integer solution by more than $\alpha$. Given the definition of Algorithm \ref{alg:midr}, it is easy to conclude the following lemma.

\begin{lemma}[\cite{DRY11}] \label{lem:approx}
If $r$ is an $\alpha$-approximate rounding scheme for CPP($\V$), then Algorithm
\ref{alg:midr} is an $\alpha$-approximation algorithm for CPP($\V$).
\end{lemma}

For reasons outlined in \cite{DRY11}, implementing Algorithm \ref{alg:midr} efficiently is impossible for most rounding schemes $r$ in the literature. To get around this difficulty, they advocate designing rounding schemes that render \eqref{lp:absorbrounding} a convex optimization problem.

\begin{definition}
  Consider a randomized rounding scheme $r: \P \to dist(\S)$. We say $r$ is a \emph{convex rounding scheme} for CPP($\V$) if, whenever $v_i \in \V$ for all $i$, the objective $\ex_{S \sim r(x)}[\sum_i v_i(S)]$ is a concave function of $x$.
\end{definition}

\begin{lemma}\label{lem:convexopt}
When $r$ is a convex rounding scheme for CPP($\V$) ,
\eqref{lp:absorbrounding} is a convex optimization problem for each instance of CPP($\V$).
\end{lemma}

Under additional technical conditions, discussed in the context of combinatorial public projects in Appendix \ref{sec:solveconvex}, convex program \eqref{lp:absorbrounding}
can be solved efficiently (e.g., using the ellipsoid method). 
This reduces the design of a polynomial-time $\alpha$-approximate MIDR
algorithm to designing a polynomial-time $\alpha$-approximate convex
rounding scheme.

Summarizing, Lemmas
\ref{lem:MIDR}, \ref{lem:approx}, and \ref{lem:convexopt} give the
following informal theorem. 

\begin{theorem}[Informal]
If there exists an $\alpha$-approximate  convex rounding scheme for CPP($\V$), then there exists a truthful-in-expectation, polynomial-time,
$\alpha$-approximate mechanism for CPP($\V$).
\end{theorem}

\section{The Mechanism}\label{sec:CPP}
In this section, we prove the main result.

\begin{thm}\label{thm:CPPmain}
There is a $(1-1/e)$-approximate, truthful-in-expectation
mechanism for combinatorial public projects with matroid rank sum valuations in the
bounded-lottery-value oracle model, running in expected $\poly(n,m)$ time.
\end{thm}

We structure the proof of Theorem \ref{thm:CPPmain} as follows. We
define the \emph{$k$-bounded-lottery rounding scheme}, which we denote by $r_k$, in Section
\ref{sec:klotteryrounding}. We prove that $r_k$ is
$(1-1/e)$-approximate (Lemma \ref{lem:klotteryapprox}), and convex
(Lemma \ref{lem:klotteryconvex}).  Lemmas \ref{lem:MIDR},
\ref{lem:approx} and \ref{lem:klotteryapprox}, taken together, imply
that Algorithm \ref{alg:midr} when instantiated  with $r=r_k$, is a $(1-1/e)$-approximate
MIDR allocation rule. Lemma \ref{lem:klotteryconvex} reduces implementing this allocation rule  to solving a convex program.

In Appendix \ref{sec:solveconvex}, we handle the technical and
numerical issues related to solving convex programs. First, we prove
that our instantiation of Algorithm \ref{alg:midr} can be implemented
in expected polynomial-time using the ellipsoid method under a
simplifying assumption on the numerical conditioning of our convex
program (Lemma \ref{lem:solveconvexconditioned}). Then we show in Section
\ref{sec:noise} that the previous assumption can be removed by
slightly modifying our algorithm.
%at the cost of $o(1)$ in the approximation ratio of Lemma \ref{lem:approx} and 
%without otherwise
%affecting our results.   
%payments nontrivial?? Should I uncomment this statment and write up
%payment lemma? 

Finally, we prove that truth-telling VCG payments can be computed
efficiently in Lemma \ref{lem:compute_payments}.  
Taken together, these lemmas complete the proof of Theorem
\ref{thm:CPPmain}. 

\subsection{The $k$-Bounded-Lottery Rounding Scheme}
\label{sec:klotteryrounding}
We devise a rounding scheme $r_k$ that we term the \emph{$k$-bounded-lottery
  rounding scheme}.  Given a feasible solution $x$ to linear program \eqref{lp:absorbrounding}, we let distribution $r_k(x)$ be the $k$-bounded-lottery with marginals $x/k$ (and promise $\emptyset$), as defined in Section \ref{sec:lotteryval}. We make this more explicit in Algorithm \ref{alg:round}. 

\begin{algorithm}
\caption{The $k$-Bounded-Lottery Rounding Scheme $r_k$}
\label{alg:round}
\begin{algorithmic}[1]
\INPUT Fractional solution $x\in \RR^m$ with $\sum_j x_{j} \leq k$, and $0 \leq x_{j} \leq 1$ for all $j$.
\OUTPUT $S \sse [m]$ with $|S| \leq k$
\STATE For each $j \in [m]$ designate the interval $I_j =[\frac{1}{k}\sum_{j' < j} x_{j'}, \frac{1}{k}\sum_{j' \leq j} x_{j'}]$ of length $\frac{x_j}{k}$
\STATE Draw $p_1,\ldots,p_k$ independently and uniformly from $[0,1]$
\STATE Let $S = \set{j \in [m] : \set{p_1,\ldots,p_k} \intersect I_j \neq \emptyset}$
% %\STATE $S=\emptyset$
% \FOR{$t=1,\ldots,k$}
% \IF{$\sum_j x_j/k \geq p_t$}
% \STATE Let $j_t$ be the minimum index such that $\sum_{j \leq j_t} x_j/k \geq p_t$.
% \ELSE 
% \STATE Let $j_t = *$
% \ENDIF
% \ENDFOR
% \STATE $S = \set{j_1,\ldots,j_k} \sm \set{*}$ 
 \end{algorithmic}
\end{algorithm}

The $k$-bounded-lottery rounding scheme is $(1-1/e)$ approximate and convex. We prove the approximation lemma below. As for convexity, we present a simplified proof for the special case of coverage valuations in Section \ref{sec:coverageconvex}, and present the proof for MRS valuations in Section \ref{sec:mrsconvex}.

\begin{lemma}\label{lem:klotteryconvex}
 The $k$-bounded-lottery rounding scheme is convex for CPP with MRS valuations.
\end{lemma}
\begin{lemma}\label{lem:klotteryapprox}
The $k$-bounded-lottery rounding scheme is $(1-1/e)$-approximate when valuations are submodular.
\end{lemma}
\begin{proof}
  Fix $n,m,k$ and $\set{v_i}_{i=1}^n$. Let $S \sse [m]$ be a feasible solution to CPP --- i.e. $|S| \leq k$. Let $1_S$ be the vector with $1$ in indices corresponding to $S$, and $0$ otherwise. Let $T \sim r_k(1_S)$. We will first show that each element of $j \in S$ is included in $T$ with probability at least $1-1/e$. Observe that $T$ is the union of $k$ independent draws from a distribution on $[m] \union \set{*}$, where each time the probability of $j \in S$ is $1/k$.  Therefore, the probability that $j$ is included in $T$ is  $1-(1-1/k)^k \geq 1-1/e$.

Submodularity now implies that  $\ex[ v_i(T)] \geq (1-1/e) \cdot  v_i(S)$ for each player $i$ --- this was proved in many contexts: see for example \cite[Lemma 2.2]{FeigeMV07}, and the earlier related result in \cite[Proposition 2.3]{F06}. This completes the proof. 
\end{proof}

\subsection{Warmup: Convexity for Coverage Valuations} %full
%\subsection{Warm-up: Convexity for Coverage \\ Valuations} %proc
\label{sec:coverageconvex}
In this section, we prove a special case of Lemma \ref{lem:klotteryconvex} for coverage valuations. Recall that a coverage function $f$ on ground set
$[m]$ designates some set $\Y$, and $m$ subsets $A_1,\ldots,A_m \sse
\Y$, such that $f(S) = | \union_{j \in S} A_j|$.  

Fix $n,m,k$ and $\set{v_i}_{i=1}^n$. Assume that, for each player $i$, the valuation function $v_i:2^{[m]} \to \RR$  is a coverage function. We let $v(S) = \sum_i v_i(S)$ be the welfare of a solution $S$ to CPP. It is an easy observation  that the sum of coverage functions is also a coverage function. Therefore $v(S)$ is a coverage function. We let $\Y$ be a set, and $A_1,\ldots,A_m \sse \Y$, such that $v(S) = |\union_{j \in \S} A_j | $. While our proof extends easily to the case where $\Y$ is an arbitrary measure space,  we assume in this section that $\Y$ is a finite set for simplicity. 

Let $\P$ denote the polytope of fractional solutions to CPP as given in \eqref{lp:absorbrounding}. We now show that $\ex_{S \sim r_k(x)}[ v(S) ]$ is a concave function of $x$ for $x \in \P$, completing the proof of Lemma \ref{lem:klotteryconvex} for the special case of coverage valuations. Take an arbitrary $x \in \P$, and let $S \sim r_k(x)$ be a random variable. Using linearity of expectations, we can rewrite the expected welfare as follows.
\begin{align*}
    \ex[v(S)] = \ex[|\union_{j \in S} A_j|] = \sum_{\ell \in \Y} \pr[ \ell \in \union_{j \in S} A_j]
\end{align*}
Since the sum of concave functions is concave, showing that $\pr [ \ell \in \union_{j \in S} A_j]$ is concave in $x$ for each $\ell \in \Y$ suffices to complete the proof. For $\ell \in \Y$, let $T_\ell = \set{ j \in [m] : \ell \in A_j}$ be the set of projects that ``cover'' $\ell$. Let $p_1,\ldots,p_k$ and $I_1,\ldots,I_k$ be as in Algorithm \ref{alg:round}. Note that $\set{I_j}_{j=1}^m$ are disjoint sub-intervals of $[0,1]$, and $|I_j| = \frac{x_j}{k}$. We can rewrite the probability of covering $\ell$ as follows.
\begin{align*}
     \pr[ \ell \in \union_{j \in S} A_j] &= \pr[ S \intersect T_\ell \neq \emptyset ] \\
&= \pr[\set{p_1,\ldots,p_k} \intersect    \union_{j \in T_\ell} I_j \neq \emptyset] \\
&= 1- \pr[ \set{p_1,\ldots,p_k} \intersect \union_{j \in T_\ell} I_j = \emptyset] \\
&= 1- \prod_{t=1}^k \pr[ p_t \notin \union_{j \in T_\ell} I_j ]\\
&= 1- \prod_{t=1}^k (1- |\union_{j \in T_\ell} I_j|)\\
&= 1-  \left(1- \frac{\sum_{j \in T_\ell} x_j}{k} \right)^k.
\end{align*}
The final form is simply the composition of the concave function $g(y)= 1-(1-y/k)^k$ with the affine function $y \to \sum_{j \in T_\ell} x_j$. It is well known that composing a concave function with an affine function yields another concave function (see e.g. \cite{Boyd}). This completes the proof.

\subsection{Convexity for Matroid Rank Sum Valuations}
\label{sec:mrsconvex}
In this section, we will prove Lemma \ref{lem:klotteryconvex} in its full generality.  First, we recall the \emph{discrete hessian matrix}, as defined in \cite{DRY11}. 
\begin{definition}[\cite{DRY11}]
  Let $v:2^{[m]} \to \RR$ be a set function. For $S \sse [m]$, we define the discrete Hessian matrix $\H^v_S \in \RR^{m \times m}$ of $v$ at $S$ as follows:
\begin{equation}
\label{eq:discretehess}
  \H^v_S(i,j) = v(S \union \set{i,j}) -v(S \union \set{i}) - v(S \union \set{j})  + v(S) 
\end{equation} 
 for $i, j \in [m]$.
\end{definition}

It was shown in \cite{DRY11} that the discrete hessian matrices are negative semi-definite for matroid rank sum functions.
\begin{claim}[\cite{DRY11}]
\label{claim:discrete_convex}
  If $v:2^{[m]} \to \RR^+$ is a matroid rank sum function, then $\H^v_S$ is negative semi-definite for each $S \sse [m]$.
\end{claim}

 We now return to Lemma \ref{lem:klotteryconvex}. Fix $n$ and $m$. For each cardinality bound $k \in [m]$, let $\P_k$ denote the polytope of fractional solutions to CPP as given in \eqref{lp:absorbrounding}.  For a set of MRS valuations $v_1,\ldots,v_n$, we observe that the social welfare $v(S)=\sum_{i=1}^n v_i(S)$ is --- by the (obvious) fact that the sum of MRS functions is an MRS function --- also an MRS function. Therefore, we will prove Lemma \ref{lem:klotteryconvex} by showing that, for each $k \in [m]$ and MRS function $v:2^{[m]} \to \RR$, the following function of $x \in \P_k$ is concave in $x$.
\begin{equation}
\begin{split}
\label{eq:G}
 G^v_k(x) &= \ex_{S \sim r_k(x)}[ v(S)]\\
 &= \sum_{S \sse [m]} v(S) \pr[r_k(x)=S] 
\end{split}
\end{equation}

 We use techniques from combinatorics to write $\pr[r_k(x)=S]$ in a form that will be easier to work with. For $T \sse [m]$, we use $x_T$ as short-hand for $\sum_{j \in T} x_j$, and $\bar{T}$ as short-hand for $[m] \sm T$.
\begin{claim}\label{claim:prob_rewrite}
For each $k \in [m]$, $x \in \P_k$, and $S \sse [m]$
\begin{equation}
  \label{eq:prob_rewrite}
  \pr[r_k(x) = S] = -1^{|S|} \sum_{R \sse S} -1^{|R|} \left(1- \frac{x_{\bar{R}}}{k}\right)^k
\end{equation}  
\end{claim}
\begin{proof}
It is easy to see that  $\pr[ r_k(x) = S ]$ is equal to:
\begin{align}\label{eq:prob_rewriting1}
  \pr[ r_k(x) \sse S] - \pr[ \bigvee_{j \in S} r_k(x) \sse S \sm \set{j} ] 
\end{align}
Using the inclusion-exclusion principle, we can rewrite \eqref{eq:prob_rewriting1} as follows:
\begin{align}
  \label{eq:prob_rewriting2}
  \pr[ r_k(x) \sse S] - \sum_{\emptyset \neq T \sse S} -1^{|T|-1} \pr[r_k(x) \sse S \sm T] 
\end{align}
Letting $R= S \sm T$ in \eqref{eq:prob_rewriting2}, we get
\begin{align}\label{eq:prob_rewriting3}
  \pr [ r_k(x) \sse S] - \sum_{R \subsetneq S} -1^{|S| - |R| - 1} \pr[ r_k(x) \sse R] 
\end{align}
We can easily simplify \eqref{eq:prob_rewriting3} to conclude that
\begin{align}
  \label{eq:prob_rewriting4}
 \pr[r_k(x)=S]=\sum_{R \sse S} -1^{|S| - |R|} \pr[ r_k(x) \sse R]
\end{align}
Next, we observe that the expression $\pr[ r_k(x) \sse R]$ can be expressed as a simple closed form in $x$. Let $p_1,\ldots,p_k$ and $I_1,\ldots,I_m$ be as in Algorithm \ref{alg:round}. The event $r_k(x) \sse R$ occurs exactly when none of $p_1,\ldots,p_k$ land in the intervals corresponding to projects $\bar{R}$. Recalling that the interval $I_j$ of project $j$ has length $x_j /k$, we get that the probability of any particular $p_t$  falling in $\union_{j \in \bar{R} } I_j$ is exactly $x_{\bar{R}} /k$. Therefore, by the independence of the variables $p_1,\ldots,p_k$, we get that
\begin{align}\label{eq:prob_rewriting5}
  \pr[r_k(x) \sse R] = \left(1- \frac{x_{\bar{R}}}{k}\right)^k
\end{align}
Combining \eqref{eq:prob_rewriting4} and \eqref{eq:prob_rewriting5} completes the proof.
\end{proof}

Building on Claim \ref{claim:prob_rewrite}, we now express the Hessian matrix of $G^v_k$ as a non-negative weighted sum of discrete Hessian matrices of $v$. We note that when $x \in \P_k$, it is easy to verify that  $\frac{k-2}{k}\cdot x \in \P_{k-2}$, and therefore \eqref{eq:G_grad} is well-defined.

\begin{claim}
\label{claim:G_grad}
  For each $k \in [m]$, $x \in \P_k$, and $v: 2^{[m]} \to \RR$, we have
\begin{align}\label{eq:G_grad}
\grad^2 G^v_k(x) = \frac{k-1}{k} \sum_{S \sse [m]} \pr\left[ r_{k-2}\left(\frac{k-2}{k}\cdot x \right) = S \right] \H^v_S
\end{align}
\end{claim}
\begin{proof}
 Fix $i,j \in [m]$, possibly with $i = j$. 
%We will show that:
%\[
%\frac{\partial^2 G^v_k(x)}{\partial x_i \partial x_j} = \frac{k-1}{k} \sum_{S \sse [m]} \pr\left[ r_{k-2}\left(\frac{k-2}{k}x \right) = S \right] \H^v_S(i,j) 
%\]
We work with $G^v_k$ as defined in Equation \eqref{eq:G}, and plug in expression \eqref{eq:prob_rewrite}. 
\begin{align*}
 G^v_k(x)  &= \sum_{S \sse [m]} v(S) \cdot -1^{|S|} \sum_{R \sse S} -1^{|R|} \left(1- \frac{x_{\bar{R}}}{k}\right)^k
\end{align*}
Differentiating with respect to $x_i$ and $x_j$ gives:
\begin{align*}
  \frac{\partial^2 G^v_k(x)}{\partial x_i \partial x_j} = \frac{k-1}{k} 
\sum_{S \sse [m]}  v(S) \cdot -1^{|S|}  \sum_{R \sse S\sm\set{i,j}}  -1^{|R|} \left(1- \frac{x_{\bar{R}}}{k}\right)^{k-2}
\end{align*}
We group the terms by projecting $S$ onto $[m] \sm \set{i,j}$, and then we simplify the resulting expression.
\begin{align}\label{eq:Gintermediate}
  \frac{\partial^2 G^v_k(x)}{\partial x_i \partial x_j} = &\frac{k-1}{k} \hspace{-0.2cm} 
\sum_{S \sse [m]\sm \set{i,j}} \hspace{-0.6cm} -1^{|S|} \sum_{R \sse S} \hspace{-0.1cm} -1^{|R|} \left(1- \frac{x_{\bar{R}}}{k}\right)^{k-2}  (v(S) - v(S \union \set{i}) - v(S \union \set{j}) +v(S \union \set{i,j})) \notag\\
 = &\frac{k-1}{k} 
\sum_{S \sse [m]} -1^{|S|} \sum_{R \sse S} \hspace{-0.1cm} -1^{|R|} \left(1- \frac{x_{\bar{R}}}{k}\right)^{k-2}  (v(S) - v(S \union \set{i}) - v(S \union \set{j}) +v(S \union \set{i,j})) \notag\\
= &\frac{k-1}{k} 
\sum_{S \sse [m]} -1^{|S|} \sum_{R \sse S} -1^{|R|} \left(1- \frac{x_{\bar{R}}}{k}\right)^{k-2} \H^v_S(i,j) 
\end{align}
The second equality follows from the fact that $v(S) - v(S \union \set{i}) - v(S \union \set{j}) + v(S \union \set{i,j}) = 0$ when $S$ includes either of $i$ and $j$.  The last equality follows by definition of $\H^v_S$.

 Invoking Claim \ref{claim:prob_rewrite} with $k'=k-2$ and $x'=\frac{k-2}{k} \cdot x$, and plugging the resulting expression into into \eqref{eq:Gintermediate}, we conclude that
\begin{align*}
  \frac{\partial^2 G^v_k(x)}{\partial x_i \partial x_j} = &\frac{k-1}{k} 
\sum_{S \sse [m]} \pr\left[r_{k-2}\left(\frac{k-2}{k} \cdot x\right) = S\right] \H^v_S(i,j).
\end{align*}
\end{proof}

Claims \eqref{claim:discrete_convex} and \eqref{claim:G_grad} establish that, when $v$ is MRS and $k \in [m]$, $\grad^2 G^v_k(x)$ is a non-negative weighted sum of negative semi-definite matrices for each $x \in \P_k$. A non-negative weighted sum of negative semi-definite matrices is negative semi-definite. Therefore,  the Hessian matrix of $G^v_k$ is negative semi-definite at each $x \in \P_k$, and we conclude that $G^v_k$ is a concave function on $\P_k$. This completes the proof of Lemma \ref{lem:klotteryconvex}.

\section*{Acknowledgments} The author thanks Tim Roughgarden and Qiqi Yan for helpful discussions.

{%\small
\bibliography{cppconvex}
\bibliographystyle{plain}
}

\newpage
\appendix

\section{Solving The Convex Program}
\label{sec:solveconvex}

In this section, we overcome some technical difficulties related to the solvability of convex programs. We follow the general outline of \cite[Appendix B]{DRY11}, modifying the proofs throughout in order to handle  the additional technical difficulties specific to CPP. We show in Section \ref{sec:solveconvexapprox} that, in the bounded-lottery-value oracle model, the four conditions for ``solvability'' of convex programs, as stated in Fact \ref{fact:convex_solvability}, are easily satisfied for convex program \eqref{lp:absorbrounding} when $r=r_k$. However, an additional challenge remains: ``solving'' a convex program --- as in Definition \ref{def:rsolvable} --- returns an approximately optimal solution. Indeed the optimal solution of a convex program may be irrational in general, so this is unavoidable.

% However, a more involved challenge remains: Algorithm \ref{alg:midr}, as stated, requires an optimal solution $x^*$ to convex program \eqref{lp:CArelaxation}. Unfortunately, most convex programs cannot be solved exactly -- indeed the entries of $x^*$ may be irrational! Using the ellipsoid method to ``solve'' a convex program invariably returns an approximately
% optimal solution, whereas a maximal-in-distributional-range mechanism
% is by definition \emph{exactly} optimal.

We show how to overcome this difficulty if we settle for polynomial runtime in expectation. While the optimal solution $x^*$ of \eqref{lp:absorbrounding} cannot be computed explicitly, the random variable $r_k(x^*)$ can be sampled in expected polynomial-time. The key idea is the following: \emph{sampling the random variable $r_k(x^*)$ rarely requires precise
 knowledge of $x^*$}.  Depending on the coin
 flips of $r_k$,  we decide how accurately we need to solve
 convex program \eqref{lp:absorbrounding} in order compute $r_k(x^*)$.  Roughly speaking, we show that the probability of requiring a  $(1-\epsilon)$-approximation falls
exponentially in $\frac{1}{\epsilon}$. As a result, we can sample $r_k(x^*)$ in expected polynomial-time. We implement this plan in Section \ref{sec:wellconditioned} under the simplifying assumption that convex program \eqref{lp:absorbrounding} is \emph{well-conditioned} --- i.e. is ``sufficiently concave'' everywhere. In Section \ref{sec:noise}, we show how to remove that assumption by slightly modifying our algorithm. 

%To simplify exposition, we address these issues in  Section \ref{sec:wellconditioned} under a simplifying assumption that can be removed. In particular, we assume in that section that the objective $f(x)$ of convex program \eqref{lp:CArelaxation} is \emph{well-conditioned} -- i.e. ``sufficiently concave'' everywhere. Under this assumption, we prove that Algorithm \ref{alg:midr} can be simulated in expected polynomial-time when $r = \rp$ (Lemma \ref{lem:solveCAconditioned}). In Section \ref{sec:noise}, we show how to remove that assumption by slightly modifying our algorithm. 
%at the cost of $o(1)$ in the approximation ratio. 

\subsection{Approximating the Convex Program}
\label{sec:solveconvexapprox}
\begin{claim}
\label{claim:solveconvexapprox}
  There is an algorithm for Combinatorial Public Projects with MRS valuations in the bounded-lottery-value oracle model that takes as input an instance of the problem and an approximation parameter $\epsilon >0$, runs in $\poly(n,m,\log(1/\epsilon))$ time, and returns a $(1-\epsilon)$-approximate solution to convex program \eqref{lp:absorbrounding} when $r=r_k$.
\end{claim}
It suffices to show that the four conditions of Fact \ref{fact:convex_solvability} are satisfied in our setting. The first three are immediate from elementary combinatorial optimization (see for example \cite{schrijver}).
%: the feasible set $\R$ of convex program \eqref{lp:CArelaxation} is a matroid polytope corresponding to a matroid on a ground set of size  $mn$, and therefore: (1) Has dimension polynomial in the parameters $m$ and $n$ of the problem. (2) We can easily compute a starting ellipsoid and lower-bound on the volume of $\R$ as required. (3) There is a separation oracle for $\R$. For (2) and (3), see for example \cite{schrijver}.  
It remains to show that the first-order oracle, as defined in Fact \ref{fact:convex_solvability}, can be implemented in polynomial-time in the bounded-lottery-value oracle model. We let $f(x)$ denote the objective function of convex program \eqref{lp:absorbrounding} when $r=r_k$. This objective can, by definition, be written as follows.
\[f(x)= \ex_{S \sim r_k(x)} \left[\sum_i v_i(S)\right] =  \sum_i G^{v_i}_k(x)\]
where  $v_i$ is the valuation function of player $i$  and $G^{v_i}_k$ is as defined in \eqref{eq:G}. By definition,  $G^{v_i}_k(x)$ is the outcome of querying the bounded-lottery-value oracle of $v_i$ with bound $k$ and marginals $x/k$. Therefore, we can evaluate $f(x)$ using  $n$ bounded-lottery-value queries, one for each player. It remains to show that we can also evaluate the (multi-variate) derivative $\grad f(x)$ of  $f(x)$. Using definition \eqref{eq:G} and Claim \ref{claim:prob_rewrite}, we take the partial derivative of $G^{v_i}_k$ with respect to $x_{j}$ and simplify the resulting expression.

\begin{align}
  \label{eq:14}
  \frac{\partial G^{v_i}_k}{\partial x_{j}}(x) &=  \sum_{S \sse [m]}   -1^{|S|} v_i(S) \sum_{R \sse S \sm \set{j}}  -1^{|R|+1} \left(1- \frac{x_{\bar{R}}}{k}\right)^{k-1} \notag\\
%&=  \hspace{-0.4cm}\sum_{S \sse [m]\sm \set{j}}\hspace{-0.3cm}   (-1^{|S|} v_i(S) + -1^{|S|+1} v_i(S \union \set{j}))  \\
%&\ \ \ \ \times \sum_{R \sse S}  -1^{|R|+1} \left(1- \frac{x_{\bar{R}}}{k}\right)^{k-1}  \\
&= \hspace{-0.3cm}  \sum_{S \sse [m]\sm \set{j}} \hspace{-0.3cm}  -1^{|S|}\left( v_i(S \union \set{j}) - v_i(S)\right)   \sum_{R \sse S}  -1^{|R|} \left(1- \frac{x_{\bar{R}}}{k}\right)^{k-1}\notag  \\
&=  \sum_{S \sse [m]}  -1^{|S|}\left(v_i(S \union \set{j}) - v_i(S)\right)  \sum_{R \sse S} -1^{|R|} \left(1- \frac{x_{\bar{R}}}{k}\right)^{k-1}  \notag \\
%&= \hspace{-0.2cm}\sum_{S \sse [m]}\hspace{-0.25cm} -1^{|S|} v_i(S \union \set{j}) \hspace{-0.1cm}  \sum_{R \sse S}\hspace{-0.2cm} -1^{|R|} \left(1- \frac{\frac{k-1}{k} x_{\bar{R}}}{k-1}\right)^{k-1}  \notag\\ 
%& \ \ \ \  - \sum_{S \sse [m]} \hspace{-0.2cm} -1^{|S|} v_i(S)  \sum_{R \sse S}\hspace{-0.2cm} -1^{|R|} \left(1- \frac{\frac{k-1}{k} x_{\bar{R}}}{k-1}\right)^{k-1} \notag  \\
&= \sum_{S \sse [m]} v_i(S \union \set{j})  \pr\left[ r_{k-1}\left(\frac{k-1}{k}x\right)  = S\right]   - \sum_{S \sse [m]} v_i(S)  \pr\left[ r_{k-1}\left(\frac{k-1}{k}x\right)  = S\right]. 
\end{align}
The second equality follows by grouping the terms of the summation  by the projection of $S$ onto $[m]\sm \set{j}$. The third equality follows from the observation that $v(S \union \set{j}) - v(S)=0$ when $S$ includes $j$. The fourth equality follows by a simple re-arrangement and application of Claim \ref{claim:prob_rewrite}.

Inspect the final form \eqref{eq:14} in light of the definition of  bounded-lottery-value oracles (Definition \ref{def:bounded_lottery_oracle}) and the definition of  $r_k$ (Section \ref{sec:klotteryrounding}). Notice that the first term is the expected value of $v_i$ over the $(k-1)$-bounded-lottery with marginals $\frac{k-1}{k}x$ and promise $\set{j}$. The second term is the expected value of $v_i$ over the same lottery without the promise. Therefore, we can evaluate $\frac{\partial G^{v_i}_k}{\partial x_{j}}(x)$ using two queries to the bounded-lottery-value oracle of player $i$.   This completes the proof of Claim \ref{claim:solveconvexapprox}.

\subsection{The Well-Conditioned Case}
\label{sec:wellconditioned}

In this section, we make the following simplifying assumption:  The objective function $f(x)$ of convex program \eqref{lp:absorbrounding} with $r=r_k$, when restricted to any line in the feasible set $\P$, has a  second derivative of magnitude at least $\lambda=\frac{\sum_{i=1}^n v_i([m])}{ 2^{\poly(n,m)}}$ everywhere, where the polynomial in the denominator may be arbitrary. This is equivalent to requiring that every eigenvalue of the Hessian matrix of $f(x)$  has magnitude at least $\lambda$ when evaluated at any point in $\P$. Under this assumption, we prove Lemma \ref{lem:solveconvexconditioned}.
\begin{lemma}\label{lem:solveconvexconditioned}
  Assume the magnitude of the second derivative of $f(x)$ is at least $ \lambda=\frac{\sum_{i=1}^n v_i([m])}{ 2^{\poly(n,m)}}$ everywhere. Algorithm \ref{alg:midr}, instantiated with $r=r_k$, can be simulated in time polynomial in $n$ and $m$ in expectation.
\end{lemma}

%We observe that, in order to simulate Algorithm \eqref{alg:midr} when using the poisson rounding procedure of Algorithm \ref{alg:round}, we need not compute the optimal solution $x^*$ of convex program \eqref{lp:CA} explicitly. At the end of the day, it suffices to sample from $\rp(x^*)$. We notice that this can be done in expected polynoimal time. The reason is thus: For most draws of the internal random coins $\set{p_j}_j$ of rounding scheme \eqref{alg:round}, computing the allocation requires only coarse knowledge of $x^*$. Therefore, using the ellipsoid algorithm to compute an approximation $\tilde{x}$ to $x^*$ suffices in these cases. On the rare occasions where coarse estimates of $x^*$ do not suffice, we can spend extra time computing better estimates until we have the desired precision. We now make this more precise.

Let $x^*$ be the optimal solution to convex program \eqref{lp:absorbrounding} with $r=r_k$. Algorithm \ref{alg:midr} outputs a set of projects distributed as $r_k(x^*)$. The $k$-bounded-lottery rounding scheme, as described in Algorithm \,\ref{alg:round}, requires making $k$ independent decisions: for $\ell \in \set{1,\ldots,k}$, we draw $p_\ell$ uniformly from $[0,1]$ and decide which interval $I_j$, if any, $p_\ell$ falls into. In other words, we find the minimum index $j_\ell$ (if any) such that $\sum_{j \leq j_\ell} x^*_{j}/k  \geq p_\ell$. Fix $\ell$. For most realizations of $p_\ell$, we can calculate $j_\ell$ using only coarse estimates $\tilde{x}_{j}$ to $x^*_{j}$. Assume we have an \emph{estimation oracle} for $x^*$ that, on input $\delta$, returns a \emph{$\delta$-estimate} $\tilde{x}$ of $x^*$: Specifically, $\tilde{x}_{j} - x^*_{j} \leq \delta$ for each $j \in [m]$.  If $p_\ell$ falls outside the ``uncertainty zones'' of $\tilde{x}$, such as when $|p_\ell-\sum_{j' \leq j} \tilde{x}_{j'} / k| > \delta m/k$ for each $j\in[m]$, it is easy to see that we can correctly determine $j_\ell$ by using $\tilde{x}$ in lieu of $x$. The total measure of the uncertainty zones of $\tilde{x}$ is at most  $2m^2\delta$,  therefore $p_\ell$ lands outside the uncertainty zones with probability at least $1-2m^2\delta$. The following claim shows that if the estimation oracle for $x^*$ can be implemented in time polynomial in $\log(1/\delta)$, then we can simulate the $k$-bounded-lottery  rounding procedure in expected polynomial-time.

\begin{claim}
  Let $x^*$ be the optimal solution of convex program \eqref{lp:absorbrounding} with $r=r_k$. Assume access to a subroutine $B(\delta)$ that returns a $\delta$-estimate of $x^*$ in  $\poly(n,m,\log(1/\delta))$ time. Algorithm \ref{alg:midr}, instantiated with $r=r_k$, can be simulated in expected $\poly(n,m)$ time.
\end{claim}
\begin{proof}
Fix $\ell \in \set{1,\ldots,k}$. Draw $p_\ell \in [0,1]$ uniformly at random as in the $k$-bounded-lottery rounding scheme in Algorithm  \ref{alg:round}.  We will show how to find, in expected $\poly(n,m)$ time, the minimum index $j_\ell$ (if any) such that $\sum_{j \leq j_\ell} x^*_{j}/k  \geq p_\ell$. 

The algorithm proceeds as follows: Start with $\delta=\delta_0=\frac{1}{2m^2}$. Let $\tilde{x}=B(\delta)$. While  $|p_\ell -\sum_{j' \leq j} \tilde{x}_{j'}/k| \leq \delta m/k$ for some $j\in [m]$ (i.e. $p_\ell$ may fall inside an ``uncertainty zone'') do the following: let $\delta=\delta/2$, $\tilde{x}=B(\delta)$ and repeat. After the loop terminates, we have a sufficiently accurate estimate of $x^*$ to calculate $j_\ell$.

It is easy to see that the above procedure is a faithful simulation of Algorithm \eqref{alg:round} on $x^*$. It remains to bound its expected running time. Let $\delta_t=\frac{1}{2^{t+1}m^2}$ denote the value of $\delta$ at iteration $t$.  By our initial assumption, iteration $t$ takes $\poly(n,m,\log(1/\delta_t)) = \poly(n,m,\log(2^{t+1} m^2))= \poly(n,m,t)$ time. The probability this procedure does not terminate after $t$ iterations is at most $2 m^2 \delta_t = 1/2^t$. Taken together, these two facts and a simple geometric summation imply that the expected runtime is polynomial in $n$ and $m$.
\end{proof}

It remains to show that the estimation oracle $B(\delta)$ can be implemented in $\poly(n,m,\log(1/\delta))$ time. At first blush, one may expect that the ellipsoid method can be used in the usual manner here. However, there is one complication: we require an estimate $\tilde{x}$ that is close to $x^*$ \emph{in solution space} rather than in terms of objective value. Using our assumption on the curvature of $f(x)$, we will reduce finding a $\delta$-estimate of $x^*$ to finding an $1-\epsilon(\delta)$ approximate solution to convex program \eqref{lp:absorbrounding} with $r=r_k$. The dependence of $\epsilon$ on $\delta$ will be such that $\epsilon \geq \poly(\delta) / 2^{\poly(n,m)}$, thereby we can invoke Claim \ref{claim:solveconvexapprox} to deduce that $B(\delta)$ can be implemented in $\poly(n,m,\log(1/\delta))$ time.

Let $\epsilon=\epsilon(\delta) = \frac{\delta^2 \lambda}{2 \sum_i v_i([m]) }$. Plugging in the definition of $\lambda$, we deduce that $\epsilon \geq \delta^2 / 2^{\poly(n,m)}$, which is the desired dependence.  It remains to show that if $\tilde{x}$ is $(1-\epsilon)$-approximate solution to \eqref{lp:absorbrounding}, then $\tilde{x}$ is also a $\delta$-estimate of $x^*$.

 Using the fact that $f(x)$ is concave, and moreover its second derivative has magnitude at least $\lambda$, it a simple exercise to bound distance of any point $x$ from the optimal point $x^*$ in terms of its sub-optimality $f(x^*) - f(x)$, as follows: 
%Drawing a ray from $x^*$ to $x$, the value of the objective $f$ is strictly decreasing as we move from $x^*$ to $x$, and its curvature is at least $\lambda$. It is now a simple exercise to deduce that 
\begin{equation}
  \label{eq:distance_suboptimality}
f(x^*) - f(x)  \geq \frac{\lambda}{2} ||x - x^*||^2. 
\end{equation}
Assume $\tilde{x}$ is a $(1-\epsilon)$-approximate solution to \eqref{lp:absorbrounding} with $r=r_k$. Equation \eqref{eq:distance_suboptimality} implies that
\begin{align*}
  ||\tilde{x} - x^*||^2 &\leq \frac{2}{\lambda} \epsilon f(x^*) 
  = \frac{\delta^2}{\sum_i v_i([m])} f(x^*) \leq \delta^2,
\end{align*}
where the last inequality follows from the fact that $\sum_i v_i([m])$ is an upper-bound on the optimal value $f(x^*)$. Therefore, $||x - x^*|| \leq \delta$, as needed.
 % impli Using the fact that $\sum_{i=1}^n v_i[m]$ is an upper bound on $f(x^*)$, and combining the definitions of $\epsilon(\delta)$ and $\lambda$, we deduce  that $\epsilon(\delta) \geq \delta^2/ 2^{\poly(n,m)}$, as needed. If a solution $x$ is such that $f(x) \geq (1-\epsilon) f(x^*)$, then we can combine this with Equation \eqref{eq:distance_suboptimality} and the definition of  $\epsilon$ to conclude that $||x - x^*|| \leq \delta$. In other words, a $(1-\epsilon)$-approximate solution $\tilde{x}$ is also a $\delta$-estimate of $x^*$. 
%By claim \ref{claim:solveCAapprox}, we can compute such a $\tilde{x}$ in time $\poly(n,m,\log(1/\epsilon))$. To complete the proof, it suffices to show that $\log(1/\epsilon) = \poly(n,m,\log(1/\delta))$.  
This completes the proof of Lemma \ref{lem:solveconvexconditioned}.

%We assume the rounding scheme is well-conditioned -- i.e. \lambda-concave for lamba > blah --, and show that flipping a coin with probability x^*_j can be done in polynomial time. While, in degenerate cases, our poissson round scheme is not well-conditioned, we argue in next subsection that it can be slightly modified to be well-conditioned.

\subsection{Guaranteeing Good Conditioning}
\label{sec:noise}
In this section, we propose a modification $r_k^+$ of the $k$-bounded-lottery rounding scheme $r_k$.  We will argue that $r_k^+$ satisfies all the properties of $r_k$ established so far, with one exception: the approximation guarantee of Lemma \ref{lem:klotteryapprox} is reduced to $1-1/e -2^{-2mn}$. Then we will show that $r_k^+$ satisfies the curvature assumption of Lemma \ref{lem:solveconvexconditioned}, demonstrating that said assumption may be removed.  Therefore Algorithm \ref{alg:midr}, instantiated with $r=r_k^+$ for combinatorial public projects with MRS valuations in the bounded-lottery-value oracle model, is $(1-1/e - 2^{-2mn})$ approximate and can be  implemented in expected $\poly(n,m)$ time. Finally, we show in Remark \ref{rem:recoverapprox}  how to recover the $2^{-2mn}$ term to get a clean $1-1/e$ approximation ratio, as claimed in Theorem \ref{thm:CPPmain}.

Let $\mu= 2^{-2nm}$.  We define $r_k^+$ in Algorithm \ref{alg:round2}. Intuitively, $r_k^+$  first chooses a tentative set $S \sse [m]$ of projects using $r_k$. Then it cancels its choice  with  small probability $\mu$. Finally, with probability $\beta$ it chooses a random project $j^* \in [m]$ and lets $S=\set{j^*}$. $\beta$ is defined as the fraction of projects included in the original tentative choice of $S$. The motivation behind this seemingly bizarre definition of $r_k^+$ is purely technical: as we will see, it can be thought of as adding  ``concave noise'' to $r_k$.

\begin{algorithm}
\caption{Modified $k$-bounded-lottery  Rounding Scheme $r_k^+$}
\label{alg:round2}
\begin{algorithmic}[1]
\INPUT Fractional solution $x\in \RR^m$ with $\sum_j x_{j} \leq k$, and $0 \leq x_{j} \leq 1$ for all $j$.
\OUTPUT Feasible solution $S \sse [m]$ with $|S| \leq k$
\STATE Let $S = r_k(x)$
\STATE Let $\beta=\frac{ |S|}{m}$
\STATE Draw $q_1 \in [0,1]$ uniformly
\IF{$q_1 \in [0,\mu]$}
\STATE Let $S = \emptyset$
\STATE Draw $q_2 \in [0,1]$ uniformly
\IF{$q_2 \in [0,\beta]$}
\STATE Choose project $j^* \in [m]$ uniformly at random.
\STATE Let $S=\set{j^*}$
\ENDIF
\ENDIF
 \end{algorithmic}
\end{algorithm}

We can write the expected welfare $\ex_{S \sim r_k^+(x)}[\sum_i v_i(S)]$ as follows. 
\begin{align*}
 \ex_{S \sim r_k(x)} \left[ (1-\mu) \sum_i v_i(S) + \mu \beta \sum_i v_i(j^*) \right].
\end{align*}
Using linearity of expectations and the fact that $\beta$ is independent of the choice of $j^*$ to simplify the expression, we get that  $\ex_{S \sim r_k^+(x)}[\sum_i v_i(S)]$ is equal to
\begin{align*}
 (1-\mu) \ex_{S \sim r_k(x)}\left[\sum_i v_i(S)\right] + \mu \ex[\beta] \frac{\sum_{j=1}^m \sum_{i=1}^n v_i(\set{j})}{m}.
\end{align*}

Observe that $r_k$ chooses a project  $j$ with probability $1-(1-x_j/k)^k$. Therefore, the expectation of $\beta$ is $\frac{\sum_{j} 1-(1-x_j/k)^k}{m}$. This gives:
\begin{align}\label{eq:rkk_expectation}
& \ex_{S \sim r_k^+(x)}\left[\sum_i v_i(S)\right] = (1-\mu) \ex_{S \sim r_k(x)}\left[\sum_i v_i(S)\right] + \frac{\mu}{m^2}  \left(\sum_{j=1}^m \sum_{i=1}^n v_i(\set{j})\right) \left(\sum_{j=1}^m   1-(1-x_j/k)^k \right).
\end{align}

It is clear that the expected welfare when using $r=r_k^+$ is within $1-\mu = 1-2^{-2nm}$ of the expected welfare when using $r=r_k$ in the instantiation of Algorithm \ref{alg:midr}. Using Lemma \ref{lem:klotteryapprox}, we conclude that $r_k^+$ is a $(1-1/e - 2^{-2nm})$-approximate rounding scheme. Moreover, using Lemma \ref{lem:klotteryconvex}, as well as the fact that $1-(1-x_j/k)^k$ is a concave function, we conclude that $r_k^+$ is a convex rounding scheme. Therefore, this establishes the analogues of Lemmas \ref{lem:klotteryapprox} and \ref{lem:klotteryconvex} for $r_k^+$. It is elementary  to verify that our proof of Lemma \ref{lem:solveconvexconditioned} can be adapted to $r_k^+$ as well.

It remains to show that $r_k^+$ is ``sufficiently concave''. This would establish that the conditioning assumption of Section \ref{sec:wellconditioned} is unnecessary for $r_k^+$. We will show that expression \eqref{eq:rkk_expectation} is a concave function with curvature of magnitude at least $\lambda=\frac{\sum_{i=1}^n v_i([m])}{ em^2 2^{2nm}}$ everywhere. Since the curvature of concave functions is always non-positive, and moreover the curvature of the sum of two functions is the sum of their curvatures, it suffices to show that the second term of the sum \eqref{eq:rkk_expectation} has curvature of magnitude at least $\lambda$. We note that the curvature of $\sum_j \left( 1-(1-x_j/k)^k\right)$ is at least $e^{-1}$  over $x\in [0,1]^{m}$. Therefore, the curvature of the second term of \eqref{eq:rkk_expectation} is at least  \[\frac{\mu}{m^2} \left(\sum_i {v_i([m])}\right) e^{-1} = \lambda\] as needed.

\begin{remark}\label{rem:recoverapprox}
  In this section, we sacrificed $2^{-2nm}$ in the approximation ratio in order to guarantee expected polynomial runtime of our algorithm even when convex program \eqref{lp:absorbrounding} is not well-conditioned. This loss can be recovered to get a clean $1-1/e$ approximation as follows. Given our $(1-1/e-2^{-2nm})$-approximate MIDR algorithm $\A$, construct the following algorithm $\A'$: Given an instance of combinatorial public projects, $\A'$ runs $\A$ on the instance with probability $1- e2^{-2nm}$, and with the remaining probability solves the instance optimally in exponential time $O(2^{2nm})$. It was shown in \cite{DR10} that a random composition of MIDR mechanisms is MIDR, therefore $\A'$ is MIDR. The expected runtime of $A'$ is bounded by the expected runtime of $\A$ plus $e 2^{-2nm} \cdot O(2^{2nm}) =O(1)$. Finally, the expected approximation  of $A'$ is the weighted average of the approximation ratio of $\A$ and the optimal approximation ratio $1$, and is at least $(1-e2^{-2nm})(1-1/e - 2^{-2nm}) + e2^{-2nm} \geq 1-1/e$.
\end{remark}

\section{Additional Preliminaries}

\subsection{Matroid Theory}
\label{app:matroids}
In this section, we review some basics of matroid theory. For a more comprehensive reference, we refer the reader to \cite{oxley}.   

A {\em matroid} $M$ is a pair $(\X,\I)$, where $\X$ is a finite \emph{ground set}, and $\I$ is a non-empty family of subsets of $\X$  satisfying the following two properties.  (1) {\em Downward closure:} If~$S$ belongs
to~$\I$, then so do all subsets of~$S$. (2)
{\em The Exchange Property:} Whenever $T,S \in \I$ with $|T| < |S|$,
there is some $x \in S \sm T$ such that $T \union \set{x} \in \I$. Elements of $\I$ are often referred to as  the \emph{independent sets} of the matroid. Subsets of $\X$ that are not in $\I$ are often called \emph{dependent}.

We associate with matroid $M$ a set function $rank_M:2^\X \to \NN$, known as the \emph{rank function of $M$}, defined as follows: $rank_M(A)= \max_{S \in \I} |S \intersect A|$. 
Equivalently, the rank of set $A$ in matroid $M$ is the maximum size of an independent set contained in $A$.  
A set function $f$ on a ground set $\X$ is a \emph{matroid rank function} if there exists a matroid $M$ on the same ground set such that $f=rank_M$. 
Matroid rank functions are monotone ($f(S) \leq f(T)$ when $S \sse T$), normalized ($f(\emptyset)=0$), and submodular ($f(S) + f(T) \geq f(S \intersect T) + f(S \union T)$ for all $S$ and $T$).

\subsection{Convex Optimization}
\label{sec:convexoptimization}

In this section, we distill some basics of convex optimization. For more details, see  \cite{nemirovski}.

\begin{definition}
  A maximization problem is given by a set $\Pi$ of instances $(\P,c)$, where $\P$ is a subset of some euclidean space, $c: \P \to \RR$, and the goal is to maximize $c(x)$ over $x \in \P$. We say $\Pi$  is a convex maximization problem if for every $(\P,c) \in \Pi$, $\P$ is a compact convex set, and $c: \P \to \RR$ is concave. If $c: \P \to \RR^+$ for every instance of $\Pi$, we say $\Pi$ is non-negative.
\end{definition}

\begin{definition}\label{def:rsolvable}
  We say a non-negative maximization problem $\Pi$ is \emph{$R$-solvable} in polynomial time if there is an algorithm that takes as input the representation of an instance $\I=(\P,c) \in \Pi$  --- where we use $|\I|$ to denote the number of bits in the representation --- and an approximation parameter $\epsilon$, and in time $\poly(|\I|,\log(1/\epsilon))$ outputs $x \in \P$ such that $c(x) \geq (1-\epsilon) \max_{y \in \P} c(y)$.
\end{definition}

\begin{fact}\label{fact:convex_solvability}
  Consider a non-negative convex maximization problem $\Pi$. If the following are satisfied, then $\Pi$ is $R$-solvable in polynomial time using the ellipsoid method. We let $\I=(\P,c)$ denote an instance of $\Pi$, and let $m$ denote the dimension of the ambient euclidean space.
  \begin{enumerate}
  \item Polynomial Dimension: $m$ is polynomial in $|\I|$.
  \item Starting ellipsoid: There is an algorithm that computes, in time $\poly(|\I|)$, a point $c \in \RR^m$, a matrix $A \in \RR^{m \by m}$, and a number $\V \in \RR$ such that the following hold. We use $E(c,A)$ to denote the ellipsoid given by center $c$ and linear transformation $A$.
    \begin{enumerate}
    \item  $E(c,A) \supseteq  \P$
    \item $\V \leq volume(\P)$
    \item $\frac{volume(E(c,A))}{\V} \leq 2^{\poly(|\I|)}$
    \end{enumerate}
  \item Separation oracle for $\P$: There is an algorithm that takes takes input $\I$ and $x \in \RR^m$, and in time $\poly(|\I|,|x|)$ where $|x|$ denotes the size of the representation of $x$, outputs ``yes'' if $ x\in \P$, otherwise outputs $h \in \RR^m$ such that $h^T x < h^T y$ for every $y \in \P$.
  \item First order oracle for $c$: There is an algorithm that takes input $\I$ and $x \in \RR^m$, and in time $\poly(|\I|,|x|)$ outputs $c(x) \in \RR$ and $\grad c ( x) \in \RR^m$.

  \end{enumerate}

\end{fact}

\subsection{Computing Payments}

\begin{lemma}
  \label{lem:compute_payments}
Let $\A$ be an MIDR allocation rule for combinatorial public projects, and let $v_1, \ldots, v_n$ be input valuations. Assume  black-box access to $\A$, and value oracle access to $\set{v_i}_{i=1}^n$.  We can compute, with $\poly(n)$ over-head in runtime, payments $p_1,\ldots,p_n$ such that $\ex[p_i]$ equals the VCG payment of player $i$ for MIDR allocation rule $\A$ on input $v_1,\ldots,v_n$.
\end{lemma}

We note that an essentially identical lemma was proved in \cite{DRY11}. Nevertheless, we include a proof for completeness.

\begin{proof}
%  It suffices to compute a (random) payment for each player with expectation equal to his VCG payment. 

 Without loss of generality, it suffices to show how to compute $p_1$. Let ${\bf 0} : 2^{[m]} \to \RR$ be the valuation evaluating to $0$ at each bundle. Recall (see e.g. \cite{Nis07}) that the VCG payment of player $1$  is equal to 
\begin{align}\label{eq:vcgpayment}
 \ex_{T \sim \A({\bf 0},v_2,\ldots,v_n)}\left[\sum_{i=2}^n v_i(T)\right] -\ex_{S \sim \A(v_1,\ldots,v_n)}\left[\sum_{i=2}^n v_i(S)\right]. 
\end{align}

Let $S$ be a sample from $\A(v_1,\ldots,v_n)$, and let $T$ be a sample from $\A({\bf 0},v_2,\ldots,v_n)$. Let $p_1= \sum_{i=2}^n v_i(T) - \sum_{i=2}^n v_i(S)$. Using linearity of expectations, it is easy to see that the expectation of $p_1$ is equal to the expression in \eqref{eq:vcgpayment}. This completes the proof.
\end{proof}

We note that the mechanism resulting from Lemma \ref{lem:compute_payments} is individually rational in expectation, and each payment is non-negative in expectation. We leave open the question of whether it is possible to enforce individual rationality and non-negative payments for our mechanism ex-post.

\end{document}